\documentclass[12pt,a4paper]{article}
\usepackage{fullpage}
\usepackage{amssymb,amsmath,stmaryrd,amsthm}
\usepackage[mathscr]{eucal}
\usepackage{pstricks,pst-node,pst-text,pst-3d}
\usepackage{framed}
\usepackage{ifthen}

\usepackage{color}
\newenvironment{EL}[1][$\ast$]{\color{magenta}\ifthenelse{\equal{#1}{}}{}{\marginpar{#1}}}{\color{black}}
\newcommand{\ELstart}{\begin{EL}}
\newcommand{\ELend}{\end{EL}}

\theoremstyle{definition}
\newtheorem{definition}{Definition}
\newtheorem{theorem}[definition]{Theorem}
\newtheorem{proposition}[definition]{Proposition}
\newtheorem{lemma}[definition]{Lemma}
\newtheorem{corollary}[definition]{Corollary}

\newtheorem{example}[definition]{Example}
\theoremstyle{remark}
\newtheorem{remark}[definition]{Remark}

\theoremstyle{claim}
\newtheorem{claim*}{Claim}

\newcommand{\svee}{\operatornamewithlimits{\varovee}}

\def \bN {\mathbb{N}}
\def \bZ {\mathbb{Z}}
\def \1{{\mathbf{1}}}

\def \I {\mathcal{I}}
\def \1{{\mathbf{1}}}
\def \cL {\mathcal{L}}
\def \gR{{\mathfrak{R}}}
\def \gS{{\mathfrak{S}}}

\begin{document}

\title{On the poset of computation rules for nonassociative calculus}

\author{Miguel COUCEIRO$^{1}$ and Michel GRABISCH$^{2}$\thanks{Corresponding
    author. Tel (+33) 1-44-07-82-85, Fax
    (+33) 1-44-07-83-01,
    email \texttt{michel.grabisch@univ-paris1.fr}} \\
\normalsize 1. Mathematics Research Unit, FSTC, University of Luxembourg\\
\normalsize   6, rue Coudenhove-Kalergi, L-1359 Luxembourg, Luxembourg    \\
\normalsize          2. Paris School of Economics, University of Paris I\\
\normalsize          106-112, Bd de l'H\^opital, 75013 Paris, France\\
\normalsize         Email: \texttt{miguel.couceiro@uni.lu, michel.grabisch@univ-paris1.fr}}

\date{Version of \today}
\maketitle

\begin{abstract}
The symmetric maximum, denoted by $\svee$, is an extension of the usual maximum
$\vee$ operation so that 0 is the neutral element, and $-x$ is the symmetric (or
inverse) of $x$, i.e., $x\svee(-x)=0$. However, such an extension does not
preserve the associativity of $\vee$.  This fact asks for systematic ways of
parenthesing (or bracketing) terms of a sequence (with more than two arguments) when using such an
extended maximum. We refer to such systematic (predefined) ways of parenthesing as computation rules.
 
 As it turns out there are infinitely many computation rules each of which
 corresponding to a systematic way of bracketing arguments of sequences.
 Essentially, computation rules reduce to deleting terms of sequences based on
 the condition $x\svee(-x)=0$.  This observation gives raise to a quasi-order on
 the set of such computation rules: say that rule 1 is below rule 2 if for all
 sequences of numbers, rule 1 deletes more terms in the sequence than rule 2.

In this paper we present a study of this quasi-ordering of computation rules.
In particular, we show that the induced poset of all equivalence classes of
computation rules is uncountably infinite, has infinitely many maximal elements, 
has infinitely many atoms, and it embeds the powerset of natural numbers
ordered by inclusion.
\end{abstract}


\section{Introduction}\label{intro}
Among the wide variety of algebraic structures sofar studied in the realm of aggregation theory, 
only a few have been considered with nonassociative fundamental operations; see e.g. \cite{BaeFodGra04,Gra03,Gra04,Gra06,PapStaj01},
see also \cite{GraMarMesPap09} for a recent reference. 
If commutativity, distributivity, and existence of neutral element and of symmetric element, etc., 
are often abandonned, associativity remains a desirable property in order to avoid ambiguities 
when assessing the outcome of composed computations within the algebraic structure.
 However, in certain situations such nonassociative operations are both natural and necessary: this is the case of the symmetric
maximum \cite{Gra03,Gra04}. 

For a preliminary discussion, consider the set $\bN$
of nonnegative integers and the maximum operation $\vee$ defined on it. Let us try to
build on $\bZ$ an operation $\svee$ behaving like a group addition but coinciding with
$\vee$ on the positive side, that is, for every $a,b\in\bZ$, $a\svee 0=a$
(neutral element), $a\svee (-a)=0$ (symmetry), $a\svee b = a\vee b$ if $a,b\geqslant
0$. If such an operation exists, it is necessarily nonassociative as shown below:
\begin{align}
-3\svee(3\svee 2) & = -3\svee 3=0\label{eq:0}\\
(-3\svee 3)\svee 2 & = 0\svee 2 = 2.\label{eq:2}
\end{align} 
One can show \cite{Gra03} that the best definition (in the sense that it fails
associativity on the smallest possible domain) of $\svee$ is given by:
\begin{equation}\label{eq:3}
a\svee b = \left\{ \begin{array}{ll}
-(|a| \vee |b|) & \mbox{ if } b \neq -a \mbox{ and } |a| \vee |b|=-a \mbox{ or } =-b
\\
0 & \mbox{ if } b=-a \\
|a| \vee |b| & \mbox{ otherwise.}
\end{array}
\right.
\end{equation}
Except for the case $b=-a$, $a \svee b$ equals the element among the two that has the greatest absolute value.

The main problem is how to interpret this nonassociative operation when
evaluating expressions like $\svee_{i=1}^na_i$, as it was the case in
\cite{Gra04}. The solution proposed in \cite{Gra04,Gra03} was to define
\emph{computation rules}, that is, to define systematic ways of putting
parentheses so that no ambiguity occurs. 
Since we deal with commutative operations, a simple example of a computation rule
is the following: put parentheses around each pair of maximal symmetric terms. If
we apply this to our example above, this rule corresponds to
(\ref{eq:2}). Another one is to make the computation separetely on positive and
on negative terms, and to aggregate the result: $(\svee_i a^+_i)\svee(\svee_i
a_i^-)$. This corresponds to (\ref{eq:0}). 

It is easy to see that there are many possible
computation rules, but to study them, one needs
to formalize the intuitive idea of a computation rule.
The aim of this paper is twofold: to propose a formal definition of a computation rule, which
was lacking in \cite{Gra03}, and to study the set of all computation rules
endowed with a very natural ordering.
As we will see, the poset of computation rules induced by this ordering is uncountable; 
in fact, from Corollary~\ref{cor:3} below,
it follows that this poset is equimorphic (equivalent with respect to embeddability) to 
the power set of positive integers ordered by inclusion.
Moreover, we show that the poset of computation rules
has infinitely many atoms and has infinitely many maximal elements; these are completely described in
Subsections \ref{atoms} and \ref{maximals}.

 Throughout the paper, we adopt the following notation:
if $Z$ is a set of symbols, then $\cL(Z)$
denotes the language (set of words, including the empty word $\varepsilon$)
built on the alphabet $Z$.

\section{The symmetric maximum}\label{sec:symmax}

In this section we recall basic concepts and preliminary results needed hereinafter 
(for further developments see \cite{Gra04,Gra03} and \cite[\S 9.3]{GraMarMesPap09}).
However, we assume that the reader is familiar with elementary notions in the theory of ordered sets,
and refer the reader, e.g., to \cite{Blyth05,CasLecMon07,DavPrie02} for basic background.

Let $C$ be a chain endowed
with an order $\leqslant$ and least element $0$, and let $C^-:=\{-c: \, c\in C\}$  
be its dually isomorphic copy, which we refer to as its symmetric counterpart.

We define $\tilde{C}:=C\cup C^-$, and set $0=-0$.  
Since we will only consider countable sequences of elements of  $\tilde{C}$, 
without loss of generality, we may assume that $\tilde{C}=\mathbb{Z}$, or a finite symmetric interval
of it.

Let us introduce a binary operation $\svee$ on $\tilde{C}$ fulfilling the following
independent conditions:
\begin{enumerate}
\item[(I)] $\svee$ coincides with $\vee$ on $C^2$.
\item[(II)] $-a$ is the symmetric of $a$, i.e., $a\svee(-a)=0$.
\item[(III)] $-(a\svee b)= (-a)\svee(-b)$ for all $a,b\in C$.
\end{enumerate}

As observed in Section~\ref{intro}, (I) and (II) imply that $\svee$ is
not associative.
Note also that from (III), it follows that $\svee$ coincides with the minimum on
$C^-$. The following results are not difficult to verify.
\begin{proposition}
Under the conditions (I), (II) and (III) above, no operation is associative on a
larger domain than that on which the symmetric maximum defined by (\ref{eq:3}) is associative.
\end{proposition}
\begin{proposition}\label{prop:1s}
The symmetric maximum has the following properties:
\begin{enumerate}
\item $\svee$ is commutative on $\tilde{C}$.
\item $0$ is the neutral element of $\svee$.
\item $\svee$ is associative on an expression involving $a_1,\ldots,a_n\in
  \tilde{C}$, $n>2$, if and only if  $\bigvee_{i=1}^n a_i \neq -
  \bigwedge_{i=1}^na_i$.
\item $\svee$ is nondecreasing in each argument on $\tilde{C}$. 
\end{enumerate}
\end{proposition}
Property (iii) of Proposition \ref{prop:1s} will be the basis for defining computation rules.

\section{Computation rules}\label{sec:coru}
The lack of associativity of $\svee$ 
induces ambiguity when evaluating expressions like
$\svee_{i=1}^n a_i$. To overcome this difficulty, computation rules were
proposed in \cite{Gra04,Gra03}, and
 which amount to eliminating situations where nonassociativity occurs,
as characterized by property (iii) in Proposition~\ref{prop:1s}. 

Given a sequence $(a_i)_{i\in I}$ with $I\subseteq \bN$, we say that it
\emph{fulfills associativity} if either $|I|\leqslant 2$ or $\bigvee_{i\in I}a_i\neq
-\bigwedge_{i\in I}a_i$. Hence $\svee_{i\in I}a_i$ is well-defined if and only
if $(a_i)_{i\in I}$ fulfills associativity.
Informally speaking, a computation rule is a systematic (predefined) way to delete symbols in
a sequence in order to make it associative, provided that this corresponds to some
arrangement of parentheses. 
\begin{example}
Consider the following sequence in $\bZ$: $3,2,1,0,-2,-3,-3$. A possible way to make the
sequence associative is to delete $3,-3$, which corresponds to the arrangement
\[
(3\svee -3 )\svee(-3\svee 2 \svee  -2 \svee 1 \svee 0) = -3.
\]
Another possibility is to delete all occurrences of maximal symmetric symbols, that
is, first $3,-3$ then $2,-2$, wich corresponds to:
\[
(3\svee ( -3 \svee  -3 ))\svee (2\svee -2 )\svee 1\svee  0 = 1.
\]
Even though deleting the 3 makes this sequence associative, it does not
correspond to any arrangement of parentheses.
\end{example}

In this section we reassemble these ideas and propose a formalism where the intuitive idea
of a computation rule is made precise, and show that our formalization fulfills our initial requirements. 
  
Since 0 is the neutral element of $\svee$, we deal with
sequences (words) built on $Z:=\tilde{C}\setminus\{0\}$, including the empty
sequence $\varepsilon$. Hence, we consider the language
$\cL(Z)$. Nonempty words are denoted by $\sigma=(a_i)_{i\in I}$, where $I$
is a finite index set.

We are interested in computing expressions $\svee_{i\in I}a_i$ unambiguously.
Since $\svee$ is commutative, the order of symbols in the word does not matter,
and we can consider any particular ordering of the word, like the decreasing
order of the absolute values of the elements in the sequence:
\[
(1,3,-2,-3,3,1,2) \rightarrow (3,3,-3,2,-2,1,1).
\]
Hence, we do not deal with words, but with such ordered sequences. We denote by
$\mathfrak{S}$ the set of all such sequences. We introduce a convenient and
unambiguous encoding of sequences, based on two mappings. The mapping $\theta$
assigns to every $\sigma\in\mathfrak{S}$, the list of the
absolute values in $\sigma$ in decreasing order:
\[
\theta(\sigma):= (n_1,\ldots,n_q). 
\]
We assume that $\theta(\sigma)$ is always a finite sequence of
arbitrary length.
The mapping $\psi:\mathfrak{S}\rightarrow\cup_{i\in \mathbb{N}}(\mathbb{N}_0^2)^i$
is defined by:
\[
\psi((a_i)_{i\in I}) = ((p_1,m_1),\ldots,(p_q,m_q))
\]
where $p_k,m_k$ are the numbers of occurrences of the $k$-th greatest absolute
value of elements in the sequence, $p_k$ being for the positive element, and
$m_k$ for the negative one. In other words, for $\theta(\sigma)=(n_1,\ldots,n_q )$,
the sequence $\sigma$ can be rewritten after reordering as:
\[
\sigma = (\underbrace{n_1,\ldots,n_1}_{p_1\text{
    times}},\underbrace{-n_1,\ldots,-n_1}_{m_1\text{ times}},\ldots,\underbrace{n_q,\ldots,n_q}_{p_q\text{
    times}},\underbrace{-n_q,\ldots,-n_q}_{m_q\text{ times}}).
\]  
Note that no pair in $\psi(\sigma)$ can be $(0,0)$.
\begin{example}
Consider the sequence $\sigma=(1,3,-3,2,-2,-2,3,1,1,1)$. Then
\begin{align*}
\theta(\sigma) &= (3,2,1)\\
\psi(\sigma) &= ((2,1),(1,2),(4,0)).
\end{align*} 
\end{example}
Note that $\theta(\sigma)$ and $\psi(\sigma)$ uniquely determine $\sigma$.
 Also, saying that $\sigma$ fulfills associativity means
that either $p_1$ or $m_1$ is 0. We denote by $\mathfrak{S}_0$ the set of
sequences which do \emph{not} fulfill associativity.
\begin{definition}\label{def:elem}
There exist five elementary rules $\rho_i:\gS\rightarrow\gS$,
defined as follows. For any sequence
$\sigma$ with $\psi(\sigma)=((p_k,m_k)_{k=1,\ldots,q})$:
\begin{enumerate}
\item elementary rule $\rho_1$: if $p_1>1$ and $m_1>0$, the number $p_1$ is
  changed into $p_1=1$;
\item elementary rule $\rho_2$: if $m_1>1$ and $p_1>0$, the number $m_1$ is
  changed into $m_1=1$;
\item elementary rule $\rho_3$: if $p_1>0$,
  $m_1>0$, the pair $(p_1,m_1)$ is changed into $(p_1-c,m_1-c)$, where
  $c:=p_1\wedge m_1$. If this results in the pair (0,0), then this pair is
  deleted, and all subsequent pairs $(p_k,m_k)$, $k=2,3,\ldots$,  are renumbered as $(p_{k-1},m_{k-1})$.
\item elementary rule $\rho_4$: if $p_1>0$,
  $m_1>0$, and if $p_2>0$, the number
  $p_2$ is changed into $p_2=0$. If this results in the pair (0,0), then this pair is
  deleted, and all subsequent pairs $(p_k,m_k)$, $k=3,4,\ldots$,  are renumbered as $(p_{k-1},m_{k-1})$.
\item elementary rule $\rho_5$: if $p_1>0$,
  $m_1>0$, and if $m_2>0$, the number
  $m_2$ is changed into $m_2=0$. If this results in the pair (0,0), then this pair is
  deleted, and all subsequent pairs $(p_k,m_k)$, $k=3,4,\ldots$,  are renumbered as $(p_{k-1},m_{k-1})$.
\end{enumerate}
Rules $\rho_1,\ldots,\rho_5$ have no action (i.e., $\rho_i(\sigma)=\sigma$) if
the conditions of application are not satisfied.  
\end{definition}

We define the (computation) alphabet as $\Psi:=\{\rho_1,\rho_2,\rho_3,\rho_4,\rho_5\}$. 
\begin{definition}
A \emph{computation rule} $R$ is any word built on $\Psi$, i.e.,
$R\in\cL(\Psi)$. We say that $R$ is a \emph{well-formed computation rule}
  (\emph{w.f.c.r.}) if for any sequence $\sigma\in\gS$ we have
$R(\sigma)\in\gS\setminus\gS_0$. We denote by $\gR$ the set of well-formed
computation rules.
\end{definition}
For example, $\rho_2\rho_3\rho_1$, $\rho_4^*\rho_1$,
$(\rho_1\rho_3)^*(\rho_4\rho_5)^*$ are computation rules, where as usual
$w^*$ denotes the infinite concatenation $wwwww\cdots$ of the word $w$ (we recall
that words are read from left to right).
Observe that only the two latter rules are well-formed.

Note that from Definition~\ref{def:elem}, we have $R(\sigma)=\sigma$ for any
rule $R$ and any sequence $\sigma$ in $\gS\setminus\gS_0$.
We give examples of w.f.c.r.'s which include those already proposed in
\cite{Gra03} (we leave to the reader the proof that they are well-formed):
\begin{enumerate}
\item $\langle \cdot\rangle_0 = \rho_3^*$,
\item $\langle \cdot\rangle_= = (\rho_1\rho_2\rho_3)^*$,
\item $\langle \cdot\rangle_-^+ = (\rho_4\rho_5)^*\langle \cdot\rangle_==(\rho_4\rho_5)^*\rho_1\rho_2\rho_3$,
\item $\langle \cdot\rangle_{pess} = (\rho_4\rho_5)^*\rho_1\rho_3$,
\item $\langle \cdot\rangle_{opt} = (\rho_4\rho_5)^*\rho_2\rho_3$,
\item $\langle \cdot\rangle_L = (\rho_1\rho_3)^*$,
\item $\langle \cdot\rangle_R = (\rho_2\rho_3)^*$.
\end{enumerate}
Note that $\langle\sigma\rangle^+_-=\varepsilon$ for all $\sigma\in\gS_0$. 

We use $\svee(R(\sigma))$ to denote the value of $\svee_{i\in I}a_i$ after applying the computation rule $R\in\gR$ to $\psi(\sigma)=\psi((a_i)_{i\in I}) $. 
To compute $\svee(R(\sigma))$, one
needs to delete symbols in the sequence $\theta(\sigma)$ exactly as they are
deleted in $\psi(\sigma)$.  We say that $R,R'\in\gR$ are \emph{equivalent},
denoted by $R\sim R'$, if for any sequence $\sigma\in\gS$ we have $\svee(R(\sigma)) = \svee(R'(\sigma))$.

The next fundamental theorem shows that our setting covers all possible ways of
putting parentheses on words in $\cL(Z)$ in order to make them associative\footnote{It is noteworthy to observe that this framework is suitable for any nonassociative operation which satisfies (ii) and (iii) of Proposition \ref{prop:1s}.}.  
\begin{theorem}\label{thm:computation rules}
Any computation rule applied to some $\sigma\in\gS$ corresponds to an
arrangement of parentheses and a permutation on $\sigma$. Conversely, any
arrangement of parentheses and permutation on some $\sigma\in\gS$ making the
sequence associative is equivalent to a computation rule applied to $\sigma$.
\end{theorem}
\begin{proof}
Let us define 5 basic rules applied on any sequence $\sigma$ with
$\psi(\sigma)=(p_k,m_k)_{k\in K}$ as follows:
\begin{enumerate}
\item basic rule ${\rho'}_1^k$,  for a given $k\in K$: if $p_k>1$, the
  number $p_k$ is changed into $p_k-1$;
\item basic rule ${\rho'}_2^k$, for a given $k\in K$: if $m_k>1$, the
  number $m_k$ is changed into $m_k-1$;
\item basic rule ${\rho'}_3^k$, for a given $k\in K$: if $p_k>0,m_k>0$,
  the pair $(p_k,m_k)$ is changed into $(p_k-1,m_k-1)$;  
\item basic rule ${\rho'}_4^k$, for $k>1$: if $p_{k}>0$, the number $p_k$
  is changed into $p_k-1$;
\item basic rule ${\rho'}_5^k$, for $k>1$: if $m_{k}>0$, the number $m_k$
  is changed into $m_k-1$,
\end{enumerate}
For all these rules, if a pair (0,0) appears, it is immediately deleted. Observe
that the elementary rules are concatenations of the above basic rules. Indeed, we have:
\[
 \rho_1 = ({\rho'}_1^1)^*,\quad  \rho_2 = ({\rho'}_2^1)^*,\quad  \rho_3 =
 ({\rho'}_3^1)^*,\quad  \rho_4 = ({\rho'}_4^2)^*,\quad  \rho_5 =
 ({\rho'}_5^2)^*. 
\]

\begin{claim*}\label{claim:1}
  Any way of parenthesing a word in $\cL(Z)$ corresponds to a word
(rule) in $\cL(\{{\rho'}_1^k,\ldots,{\rho'}_5^k\}_{k\in\bN})$, and conversely.
\end{claim*}
 
\begin{proof}[Proof of Claim~\ref{claim:1}.]
Indeed,
consider a word $w\in\cL(Z)$: parentheses are put around 2 consecutive elements, like $(a\svee b)$, where $a$ or $b$ can
be the result of a pair of parentheses too. Only three cases can occur:
\begin{enumerate}
\item either $a=b$, then $(a\svee b) = a = b$. This corresponds to basic rules
  ${\rho'}_1^k$ (if $a>0$) or
  ${\rho'}_2^k$ (if $a<0$) for a suitable $k$;
\item or $a=-b$, then $(a\svee b)= 0$. This corresponds to the basic rule
  ${\rho'}_3^k$ for a suitable $k$;
\item otherwise $|a|<|b|$ (or $|a|>|b|$). Then $(a\svee b) = b$ and this
  corresponds to the basic rules ${\rho'}^k_4$ (if $a>0$) or ${\rho'}_5^k$ (if
  $a<0$) for a suitable $k$.\qedhere
\end{enumerate}
\end{proof}
\medskip

\begin{claim*}\label{claim:2}
 Given a sequence $\sigma\in\gS_0$, for any rule $\rho$ in
$\cL(\{{\rho'}_1^k,\ldots,{\rho'}_5^k\}_{k\in\bN})$ making $\sigma$ associative, 
there exists a computation rule $R$ in $\cL(\Psi)$ such that
$\svee(\rho(\sigma)) = \svee(R(\sigma))$.
\end{claim*}
 
\begin{proof}[Proof of Claim~\ref{claim:2}.]
We have already established that any elementary rule is a particular rule in
$\cL(\{{\rho'}_1^k,\ldots,{\rho'}_5^k\}_{k\in\bN})$, and therefore this is true
also for any computation rule in $\cL(\Psi)$.

Take then any rule $\rho$ in $\cL(\{{\rho'}_1^k,\ldots,{\rho'}_5^k\}_{k\in\bN})$
making $\sigma$ associative. The result $\svee(\rho(\sigma))$ is some number in
$\sigma$, say $\delta n_k$, with $\delta =1$ or $-1$ (i.e., the $k$th
positive or negative symbol in $\theta(\sigma)$). Let us construct a
computation rule $R$ such that $\svee(R(\sigma))=\delta n_k$ as follows:
\begin{itemize}
\item Suppose $k=1,\delta=1$ (provided $p_1>1$). Then $R=\rho_2\rho_3$. For
  the case $\delta=-1$, we find $R=\rho_1\rho_3$.
\item Suppose $k>1$, $\delta=1$ (provided $p_k>0$) or $\delta=-1$ (provided
  $m_k>0$). Apply the following algorithm:
  \begin{itemize}
  \item Initialization: $R\leftarrow \varepsilon$
  \item For $i=2$ to $k-1$, Do:
    \begin{itemize}
    \item If $p_i=0$ put $R\leftarrow R\rho_5$
    \item If $m_i=0$ put $R\leftarrow R\rho_4$
    \item Otherwise put $R\leftarrow R\rho_4\rho_5$
    \end{itemize}
  \item Case $\delta=1$: if $m_k>0$, $R\leftarrow R\rho_5$. 
  \item Case $\delta=-1$: if $p_k>0$, $R\leftarrow R\rho_4$.
  \item $R\leftarrow R\rho_1\rho_2\rho_3$
  \end{itemize}
\end{itemize}
By construction, $R$ is equivalent to $\rho$ on $\sigma$, and the proof of 
the claim is now complete.
\end{proof}
Theorem~\ref{thm:computation rules} now follows from Claims~\ref{claim:1} and \ref{claim:2}.
\end{proof}

\begin{remark}
Note that a well-formed rule in
$\cL(\{{\rho'}_1^k,\ldots,{\rho'}_5^k\}_{k\in\bN})$ (i.e., making any $\sigma$
associative)  is not necessarily equivalent to a
w.f.c.r. in $\gR$. For instance, consider the well-formed rule
$\rho={\rho'}_5^3(({\rho'}_1^1)^*({\rho'}_2^1)^*{\rho'}_3^1)^*$, and apply it on
the sequences:
\[
\sigma=(2,3)(1,0)(0,1)(2,1), \quad \sigma'=(2,3)(1,1)(0,1)(2,0).
\] 
Then $\svee(\rho(\sigma)) = n_2$ and $\svee(\rho(\sigma'))=n_4$. Let us try to
build an equivalent w.f.c.r. $R\in \gR$ . Since the second pair in  $\sigma$ is the
final result, one cannot touch it. Therefore, $R$ contains only
$\rho_1,\rho_2,\rho_3$, and thus one finds $-n_3$ on $\sigma'$. Hence, compositions of 
basic rules may result in rules more general than our computation rules. However, those
rules which are not computation rules are rather artificial.
\end{remark}

Hereinafter, we will make use of the following ``factorization scheme" for computation rules. 
\begin{lemma}\label{lem:1a}
Let $R$ be a w.f.c.r. in $\gR$.
\begin{enumerate}
\item {\bf Factorization:} Rule $R$ can be factorized into a composition
\begin{equation}\label{eq:1}
R=T_1T_2\cdots T_i\cdots
\end{equation}
where each term has the form $T_i:=\omega_i\rho_1^{a_i}\rho_2^{b_i}\rho_3$, with
$\omega_i\in\cL(\{\rho_4,\rho_5\})$ (possibly empty),
and $a_i,b_i\in\{0,1\}$. 
\item {\bf Simplification:} Suppose that in \eqref{eq:1} there exists $j\in \mathbb{N}$ such
  that $\omega_j=\omega \rho^*_4$ or $\omega \rho^*_5$ for some $\omega \in
  \cL(\{\rho_4,\rho_5\})$, or that $\rho_4 $ and $
  \rho_5$ alternate infinitely many times in $\omega_j$. Let
\[
k_1=\min\{j:\, \mbox{$\omega_j=\omega \rho^*_4$ or $\omega \rho^*_5$} \}, \mbox{ and} 
\] 
\[
k_2=\min\{j:\, \mbox{$\rho_4 $ and $ \rho_5$ alternate infinitely many times in $\omega_j$} \}.
\]
\begin{itemize}
 \item If $k_1<k_2$, then $R\sim T_1\cdots T_{k_1}$.
\item Otherwise, $k_2\leqslant k_1$, and  $R\sim T_1\cdots T'_{k_2}$, where $T'_{k_2}=(\rho_4\rho_5)^*\rho_1^{a_{k_2}}\rho_2^{b_{k_2}}\rho_3$.
\end{itemize}
\end{enumerate}
\end{lemma}
\begin{proof}
Let $R$ be a w.f.c.r. Then $R$ is necessarily infinite, otherwise one can always
construct a sequence $\sigma$ such that
$R(\sigma)\not\in\gS\setminus\gS_0$. Also, $\rho_3$ necessarily belongs to $R$,
otherwise the sequence $\sigma$ with $\psi(\sigma)=(2,1)$ would not be made
associative by $R$. Therefore, the word $R$ can be cut into terms where $\rho_3$
acts as a separator, i.e., $R=R_1\rho_3R_2\rho_3\cdots$, with
$R_i\in\cL(\{\rho_1,\rho_2,\rho_4,\rho_5\})$. Now observe that $\rho_1$ and
$\rho_1^k$ are equivalent for any $k>1$, and the same holds for
$\rho_2$. Moreover, the order between $\rho_1,\rho_2$ and $\omega_i$ is
unimportant because none of these symbols can make the sequence $\sigma$
associative (i.e., the rule will not stop after applying these elementary
rules), and each of them applies on a different symbol of $\sigma$. This proves
that each term $T_i$ can be written in form (\ref{eq:1}).
 
Observe that since $R$ is infinite, there can be infinitely many factors $T_i$
or finitely many, provided one factor $T_i$ has an infinite $\omega_i$. In the
first case, there is no last factor and the proof of (i) is complete. In the
second case, it remains to prove that the last factor $T_l$ has the same form, i.e.,
it ends with $\rho_3$. Suppose on the contrary that there are elementary rules
$\rho_4,\rho_5$ after $\rho_3$. If $\sigma$ is made associative after applying
$\rho_3$, then the rule stops and the remaining $\rho_4,\rho_5$ are useless. If
not, it is because $\rho_3$ has acted on a pair $(p,p)$ with $p>0$. But if the
next pair is, say, (1,1), $\sigma$  will not be made associative by the
remaining $\rho_4,\rho_5$, contradicting the fact that $R$ is well-formed.

Let us prove (ii). Suppose first that $k_2\leqslant k_1$. Observe that any $\omega_i$
where $\rho_4,\rho_5$ alternate infinitely many times is equivalent to
$(\rho_4\rho_5)^*$. Moreover, $(\rho_4\rho_5)^*$ deletes all pairs after the
current one. Therefore, it remains only the current pair, and
$\rho_1^{a_{k_2}}\rho_2^{b_{k_2}}\rho_3$ necessarily stops on it, for any value
of $a_{k_2},b_{k_2}$. 

Suppose now that $k_1<k_2$, and $\omega_{k_1}=\omega\rho_4^*$ (the other case is
similar). Then $\rho^*_4$ deletes all pairs after the current pair of the
form $(p',0)$, and stops at the first pair of the form $(p',m')$ with $m'>0$,
which is transformed into $(0,m')$. Then
$\rho_1^{a_{k_1}}\rho_2^{b_{k_1}}\rho_3$ makes the current pair either of the
form $(0,0)$, or $(p,0)$ or $(0,m)$. In the two last cases, $R$ stops. In the
first case, the current term is deleted, and the next pair encountered is
$(0,m')$, where the rule stops. The proof of (ii) is complete.
\end{proof}

\begin{remark} Note that (ii) of Lemma \ref{lem:1a} does not refer to
every $\omega$ containing a $\rho_4^*$ or a $\rho_5^*$. For instance, 
if $\omega=\rho_5\rho_4^*\rho_5$, then the subsequent terms of $R$ are relevant. 
  \end{remark}

\begin{remark}\label{rem:2} If $T_i:=\omega_i\rho_1^{a_i}\rho_2^{b_i}\rho_3$, where
$\omega_i\neq (\rho_4\rho_5)^*, \omega\rho_4^*, \omega\rho_5^*$, for any
$\omega \in  \cL(\{\rho_4,\rho_5\})$, then there is $\sigma_i$ such that 
$T_i(\sigma_i)=\varepsilon$ and $T_i(\sigma_i\sigma)=\sigma$ for every $\sigma \in\gS$. 
  \end{remark}

We refer to the compositions given in (ii) as \emph{factorized irredundant forms} of computation rules. 
For instance,
$\langle\cdot\rangle_-^+$ can be factorized into two equivalent compositions
\[
\langle\cdot\rangle_-^+=(\rho_4\rho_5)^*(\rho_1\rho_2\rho_3)^* = (\rho_4\rho_5)^*\rho_1\rho_2\rho_3.
\] 
but only the second is a factorized irredundant form.
Note that our previous examples of w.f.c.r.'s are given in factorized irredundant forms.

Now it is natural to ask whether two equivalent rules have
necessarily the same factorized irredundant form.
The next proposition shows that there is a
unique factorized irredundant form for each equivalence class of computation rules.

\begin{proposition}\label{prop:ab}
Let $T=T_1\cdots T_n$ and $T'=T'_1\cdots T'_m$ be two rules in factorized
irredundant form, where $n,m$ may be infinite. Then $T\sim T'$ if and only if $n=m$ and 
for every $1\leqslant i\leqslant n$, $T_i=T'_i$.
\end{proposition} 
\begin{proof}
Clearly, the conditions are sufficient. So let us prove that they are also necessary.

First, we show that $n=m$. For a contradiction, suppose that $n\neq m$, say $n<m$.
In particular, for every $j<m$, $\omega'_j$ is not of the $\omega \rho^*_4$ nor
$\omega \rho^*_5$ form for any $\omega \in \cL(\{\rho_4,\rho_5\})$,
and $\omega'_j\neq (\rho_4\rho_5)^*$.

Note that $(a_1,b_1)=(a'_1,b'_1)$, otherwise $T\not\sim T'$ (just consider $(2,1)$, $(2,2)$, or $(1,2)$).
Thus, to verify that $T_1=T'_1$, it suffices to show that $\omega_1=\omega'_1$. Let $p$ and $p'$ be the
number of times that $\rho_4 $ and $ \rho_5$ alternate in $\omega_1$ and $\omega'_1$, respectively.
It is easy to see that $p=p'$ (just consider sequences of the form $(1,1)(1,0)^a(0,1)[(1,0)(0,1)]^q(1,0)(0,1)^b$, for suitable $a,b\in \{0,1\}$
and $q\in \mathbb{N}$). Moreover, either both start with $\rho_4$ or both start with $\rho_5$ (just consider strings of the form 
$(1,1)[(1,0)(0,1)]^p$).

So suppose both start with $\rho_4$ and $p=2t-1$ (the case $p=2t$ is similar), say 
$$\omega_1=\rho_4^{l_1}\rho_5^{r_1}\cdots \rho_4^{l_{t}}\rho_5^{r_t}\, \mbox{ and }\, \omega'_1=\rho_4^{l'_1}\rho_5^{r'_1}\cdots \rho_4^{l'_t}\rho_5^{r'_t}$$
where $r_t,r'_t\neq 0$, and let $k=\min\{j:l_j\neq l'_j \, \mbox{or} \, r_j\neq r'_j\}$, say $l_k< l'_k$ (the other cases are dealt with similarly).
Then, for 
\[
\sigma=(1,1)[(1,0)(0,1)]^{k-1}(1,0)^{l_k+1}(0,1)[(1,0)(0,1)]^{t-k}
\]
$\svee(T(\sigma))> \svee(T'(\sigma))$,
which contradicts $T\sim T'$.
Hence, $\omega_1=\omega'_1$, and we conclude $T_1=T'_1$. In fact,
following exactly the same steps, one can verify that $T_i=T'_i$, for every $i<n$.

Now, as in the case above $(a_n,b_n)=(a'_n,b'_n)$, otherwise $T\not\sim T'$.
Moreover, by assumption, we have that $\omega_n=\omega \rho^*_4$ or $\omega \rho^*_5$ for some $\omega \in \cL(\{\rho_4,\rho_5\})\setminus \{(\rho_4\rho_5)^*\}$,
or that $\omega_n=(\rho_4\rho_5)^*$. 
Since $\omega'_n\neq (\rho_4\rho_5)^*$, $\omega_n\neq (\rho_4\rho_5)^*$. Hence, $\rho_4 $ and $ \rho_5$
must alternate the same number of times, say  
\[
\omega_n=\rho_4^{l_1}\rho_5^{r_1}\cdots \rho_4^{l_{t}}\rho_5^{r_t}\, \mbox{ and }\, \omega'_n=\rho_4^{l'_1}\rho_5^{r'_1}\cdots \rho_4^{l'_t}\rho_5^{r'_t},
\] 
 where either $l_t=*\neq l'_t$ and $r_t=r'_t=0$, or $r_t=*\neq r'_t$ and $l_t, l'_t>0$. Without loss of generality, suppose that the latter holds.
Then, for 
\[
\sigma=(1,1)[(1,0)(0,1)]^{n-1}(1,0)(0,1)^{r'_t+1}\,
\]
$ \svee(T(\sigma))<  \svee(T'(\sigma)),$  again a contradiction.
Using Lemma~\ref{lem:1a} (ii), we see that all possible cases have been considered and, since each leads to a contradiction, we have
 $n=m$.

Now, by making use of (concatenations of) sequences of the form 
\[
(1,1)(1,0)^a(0,1)[(1,0)(0,1)]^q(1,0)(0,1)^b(2,1)^{c}(2,2)^d(1,2)^e,
\]
 if both have infinitely many terms 
$T_i, T'_i$, then $T_i= T'_i$ for every $i\in \mathbb{N}$, and 
if $T$ and $T'$ have  the same (finite) number of terms, say $n$, then $T_i=T'_i$ for every $i<n$, and $(a_n,b_n)=(a'_n,b'_n)$. 

Thus, to complete the proof it remains to show that in the latter case, we have $\omega_n=\omega'_n$; in fact, both $\omega_n$
and $\omega'_n$ are $(\rho_4\rho_5)^*$,
or $\omega \rho^*_4$ or $\omega \rho^*_5$ for some $\omega \in \cL(\{\rho_4,\rho_5\})\setminus \{(\rho_4\rho_5)^*\}$.

For the sake of a contradiction, suppose first that $\omega_n=(\rho_4\rho_5)^*$ but $\omega'_n=\omega \rho^*_4$ or $\omega'_n=\omega \rho^*_5$ where $\rho_4$ and $\rho_5$ alternate
finitely many times in $\omega$, say $p$ times. Then 
\[
 \svee(T(\sigma(1,1)[(1,0)(0,1)]^{p+2}))<  \svee(T'(\sigma(1,1)[(1,0)(0,1)]^{p+2})),
\]
where $\sigma$ is the concatenation of the sequences $\sigma_i$, $1\leqslant i<n$, given in Remark~\ref{rem:2}. 

Now suppose that $\omega_n=\omega \rho^*_4$ and $\omega'_n=\omega' \rho^*_5$ where $\rho_4$ and $\rho_5$ alternate
finitely many times in $\omega$ and $\omega'$, say $p$ and $p'$ times, respectively. (The remaining cases can be dealt with similarly.)
Without loss of generality, suppose that $p\leqslant p'$, then taking $a$ as the ceiling of $\frac{p}{2}$ we have 
\begin{itemize}
 \item 
 $ \svee(T(\sigma(1,1)[(1,0)(0,1)]^{a}(0,1)))>  \svee(T(\sigma(1,1)[(1,0)(0,1)]^{a}(0,1))),$
  if $\omega $ and $\omega'$ both start with $\rho_4$ or $\rho_5$, or
 $\omega $ and $\omega'$ start with $\rho_5$ and $\rho_4$, respectively, and
\item 
$ \svee(T(\sigma(1,1)(0,1)[(1,0)(0,1)]^{a}(0,1)))> \svee(T(\sigma(1,1)(0,1)[(1,0)(0,1)]^{a}(0,1))),$
 if $\omega $ and $\omega'$ start with $\rho_4$ and $\rho_5$, respectively,
\end{itemize}
where again $\sigma$ is the concatenation of the sequences $\sigma_i$, $1\leqslant i<n$, given in Remark~\ref{rem:2}. 

Since both cases yield the desired contradiction, the proof is now complete.
\end{proof}


\section{The poset $(\gR/_\sim,\leqslant)$ of computation rules}\label{sec:struc}
 The above considerations allow us to focus on the quotient $\gR/_\sim$ of
equivalence classes rather than on the whole set of w.f.c.r.'s.
Moreover, by making use of Lemma \ref{lem:1a}, we can focus
on factorized irredundant forms. 

We consider the following order $\leqslant$ on $\gR/_\sim$ which was introduced in  \cite{Gra03}.  
Let $R,R'$ be two computation
rules in $\gR/_\sim$ and, for each sequence $\sigma=(a_i)_{i\in I}$, let
 $J_\sigma$ and $J'_\sigma$, $J_\sigma,J'_\sigma\subseteq I$, 
 be the sets of indices of the terms in $\sigma$ deleted by $R$ and $R'$,
respectively. Then, we write $R\leqslant R'$ if for all sequences $\sigma\in\gS$ we have
$J_\sigma\supseteq J'_\sigma$. To simplify our exposition, 
we use $R(\sigma)\sqsubseteq R'(\sigma)$ to denote the fact that $J_\sigma\supseteq J'_\sigma$. 
If $J_\sigma= J'_\sigma$, then we simply write $R(\sigma)= R'(\sigma)$.
Moreover, we may adopt the same notation to arbitrary substrings of w.f.c.r.'s.

It is easy to verify that
$\leqslant$ is reflexive and transitive (but,  as we will see, not linear).  
Also, it is antisymmetric: if two rules
$R,R'$ delete exactly the same terms, i.e., $R\leqslant R'$ and $R'\leqslant R$, then they are
equivalent. Conversely, it follows from
Proposition~\ref{prop:ab} that if two rules are equivalent, then they have the same
factorized irredundant form, therefore $R\leqslant R'$ and $R'\leqslant R$.
Thus,  $(\gR/_\sim,\leqslant)$ is a poset (partially ordered set). 
In what follows, we make no distinction between w.f.c.r.'s and the elements of $\gR/_\sim$ which will be always written in the factorized irredundant form.

\subsection{Preliminary results}

In the sequel, let $\omega, \omega'\in\cL(\{\rho_4,\rho_5 \})$, 
and $a,b, c,d\in\{0,1\}$.

\begin{lemma}\label{lem:4}
Let $T,T' \in\gR/_\sim$. If $T\geqslant T'$, then 
$\omega\rho_1^a \rho_2^b \rho_3T\geqslant \omega\rho_1^a \rho_2^b \rho_3T'$.
Moreover, if $\rho_4$ and $\rho_5$ alternate finitely many times in $\omega$,
then $\omega\rho_1^a \rho_2^b \rho_3T> \omega\rho_1^a \rho_2^b \rho_3T'$ 
(resp. $\omega\rho_1^a \rho_2^b \rho_3T\parallel \omega\rho_1^a \rho_2^b \rho_3T'$)
if and only if $T> T'$ (resp.  $T\parallel T'$).
\end{lemma}

\begin{proof}
From Lemma~\ref{lem:1a} (ii), if $\rho_4$ and $\rho_5$ alternate infinitely many times in $\omega$, then
\[
\omega\rho_1^a \rho_2^b \rho_3T=(\rho_4\rho_5)^*\rho_1^a \rho_2^b \rho_3 =\omega\rho_1^a \rho_2^b \rho_3T'.
\]
So we may assume that $\omega:=\rho_4^{a_1}\rho_5^{b_1}\cdots\rho_4^{a_n}\rho_5^{b_n}$, with
$a_i,b_i\in \bN \cup \{*\}$. We assume also that $a_i\neq 0$ for $2\leqslant
i\leqslant n$, and $b_i\neq 0$ for $1\leqslant n-1$. We treat the case $a_1\neq 0$, $b_n\neq 0$, the remaining cases follow similarly.

To see that $\omega\rho_1^a \rho_2^b \rho_3T\geqslant \omega\rho_1^a \rho_2^b \rho_3T'$, just note that
 for every string $\gamma =(p_1,m_1)\cdots (p_k,m_k)$, $(p_1,m_1)\geq (1,1)$, we have
\begin{eqnarray*}
\omega\rho_1^a \rho_2^b \rho_3T(\gamma)&=&(\rho_1^a \rho_2^b \rho_3)(p_1,m_1)T((p'_1,m'_1)\cdots (p'_{k'},m'_{k'}))\\
 &\sqsupseteq &
(\rho_1^a \rho_2^b \rho_3)(p_1,m_1)T'((p'_1,m'_1)\cdots (p'_{k'},m'_{k'}))=\omega\rho_1^a \rho_2^b \rho_3T'(\gamma).
\end{eqnarray*}
Hence, $\omega\rho_1^a \rho_2^b \rho_3T\geqslant \omega\rho_1^a \rho_2^b \rho_3T'$, and by antisymmetry, the strict inequality
occurs if and only if $T>T'$. 
Similarly, if $T\parallel T'$, then by taking 
\begin{itemize}
 \item $\sigma_>$ and $\sigma_<$ such
that  $T(\sigma_>)\sqsupset T'(\sigma_>)$
and  $T(\sigma_<)\sqsubset T'(\sigma_<)$, respectively, and 
\item $\gamma_>=\sigma\sigma_>$ and $\gamma_<=\sigma\sigma_<, $ where 
$\sigma$ is given in Remark~\ref{rem:2} (for $T_i=\omega\rho_1^a \rho_2^b \rho_3$),
\end{itemize}
we can verify that  
$\omega\rho_1^a \rho_2^b \rho_3T\parallel \omega\rho_1^a \rho_2^b \rho_3T'$.  
This completes the proof of the lemma.
\end{proof}

By repeated applications of Lemma \ref{lem:4}, we have the following corollary.

\begin{corollary}\label{cor:5}
Let $T, T'\in \gR/_\sim$, and let $R=T_1T_2\cdots T_m\in \cL(\Psi)$, where 
$T_i=\omega_i\rho_1^{a_i} \rho_2^{b_i} \rho_3$. 
If  $T\geqslant T'$, then $RT\geqslant RT'$. 
Furthermore, if $\rho_4$ and $\rho_5$ alternate finitely many times in each $\omega_i$,
then $RT> RT'$ (resp. $RT\parallel RT'$) if and only if $T> T'$ (resp. $T\parallel T'$).
\end{corollary}

\begin{remark}\label{rem:123}
In fact, by Corollary~\ref{cor:5} it follows that if $\rho_4$ and $\rho_5$ alternate finitely many times in each $\omega_i$, then $T\geqslant T'$ (resp. $T\parallel T'$) if and only if  $RT\geqslant RT'$
(resp. $RT\parallel RT'$).
\end{remark}

\begin{lemma}\label{lem:1}
Let $T,T'\in\cL(\Psi)$ such that $T\geqslant T'$. Then $\omega\rho_1^a \rho_2^b \rho_3T\geqslant  \omega\rho_1^c \rho_2^d \rho_3T'$ if and only if $(a,b)= (c,d)$ or $(c,d)=(1,1)$. 
Moreover, $\omega\rho_1^a \rho_2^b \rho_3T> \omega\rho_1^c \rho_2^d \rho_3T'$ if and only if 
$(a,b)\neq (c,d)=(1,1)$. 
\end{lemma}
\begin{proof}  To see that the condition in the first claim is sufficient,
  observe that if $(a,b)=(c,d)$, then by Lemma~\ref{lem:4} $\omega\rho_1^a
  \rho_2^b \rho_3T\geqslant \omega\rho_1^c \rho_2^d \rho_3T'$. If $(c,d)=(1,1)$,
  then for every nonassociative string $\sigma=(p_1,m_1)\cdots
  (p_k,m_k)$ (i.e., $(p_1,m_1)\geqslant (1,1))$ we have
\begin{eqnarray*}
 \omega\rho_1^a \rho_2^b \rho_3T(\sigma)&=&(\rho_1^a \rho_2^b \rho_3)(p_1,m_1)T((p'_1,m'_1)\cdots (p'_{k'},m'_{k'})\nonumber\\
&\sqsupseteq & T'((p'_1,m'_1)\cdots (p'_{k'},m'_{k'}))=\omega\rho_1^c \rho_2^d \rho_3T'(\sigma),\label{eq:ab}
\end{eqnarray*}
 and hence $\omega\rho_1^a \rho_2^b \rho_3T\geqslant  \omega\rho_1^c \rho_2^d \rho_3T'$. 
Moreover, if  $(a,b)\neq(c,d)=(1,1)$, then by considering
 $(2,1)$ if $(a,b)$ equals $(0,1)$ or $(0,0)$, and
 $(1,2)$ if $(a,b)$ equals $(1,0)$,
one can easily verify that $\omega\rho_1^a \rho_2^b \rho_3T> \omega\rho_1^c \rho_2^d \rho_3T'$, 
thus showing that the condition of the second claim is also sufficient.

To verify that the conditions in the first and second claims are also necessary, it suffices to show that if 
$(a,b), (c,d)\neq(1,1)$ and $(a,b)\neq (c,d)$, then  $\omega\rho_1^a \rho_2^b \rho_3T\parallel \omega\rho_1^c \rho_2^d \rho_3T$.
But this fact can be easily verified by making use of the strings $(2,1), (1,2)$ or $(2,2)$, and thus the proof of the lemma is now complete.
 \end{proof}

\begin{lemma}\label{lem:2} 
 If $\rho_4$ and $\rho_5$ alternate infinitely many times in $\omega$ but not in $\omega'$, then 
$\omega\rho_1^a \rho_2^b \rho_3\sim\omega\rho_1^a \rho_2^b \rho_3T<\omega'\rho_1^a \rho_2^b \rho_3T'$,  for every $T,T'\in\cL(\Psi)$.
\end{lemma}

\begin{proof} Let $\omega=(\rho_4\rho_5)^*$ and $\omega':=\rho_4^{a_1}\rho_5^{b_1}\cdots\rho_4^{a_n}\rho_5^{b_n}$, with
$a_i,b_i\in \bN \cup \{*\}$. (By Lemma~\ref{lem:1a}, $\omega\rho_1^a \rho_2^b \rho_3\sim\omega\rho_1^a \rho_2^b \rho_3T$, for every $T\in\cL(\Psi)$.)
Then for every string $\gamma =(p_1,m_1)\cdots (p_k,m_k)$, $(p_1,m_1)\geqslant (1,1)$,
\begin{equation*}
\omega\rho_1^a \rho_2^b \rho_3(\gamma)=(\rho_1^a \rho_2^b \rho_3)(p_1,m_1)\sqsubseteq 
(\rho_1^a \rho_2^b \rho_3)(p_1,m_1)T'((p'_1,m'_1)\cdots (p'_{k'},m'_{k'}))=\omega'\rho_1^a \rho_2^b \rho_3T'(\gamma).
\end{equation*}
Hence, $\omega\rho_1^a \rho_2^b \rho_3\leqslant \omega'\rho_1^a \rho_2^b \rho_3T'$.
For $1\leqslant i\leqslant n$, let $\alpha_i=0$ (resp. $\beta_i=0$) if $a_i=0$ (resp. $b_i=0$) and  
$\alpha_i=1$ (resp. $\beta_i=1$) otherwise. By considering 
\[
\sigma=(1,1)(1,0)^{\alpha_1}(0,1)^{\beta_1}\cdots (1,0)^{\alpha_n}(0,1)^{\beta_n},
\]
 one can easily verify that
$\omega\rho_1^a \rho_2^b \rho_3<\omega'\rho_1^a \rho_2^b \rho_3T'$.
\end{proof}

\begin{lemma}\label{lem:3}
 Let $\omega:=\rho_4^{a_1}\rho_5^{b_1}\cdots\rho_4^{a_n}\rho_5^{b_n}$, $n\geq 0$, and let 
$\omega':=\rho_4^{a'_1}\rho_5^{b'_1}\cdots\rho_4^{a'_m}\rho_5^{b'_m}$, $m\geq 0$.
For $T=\omega\rho_1^a \rho_2^b \rho_3\langle \cdot \rangle^+_-$ and 
$T'=\omega'\rho_1^a \rho_2^b \rho_3\langle \cdot \rangle^+_-$, the following assertions hold:
\begin{itemize}
 \item[(i)] If $n=m=1$, then $T\parallel T'$ if and only if $b_1\neq b'_1$, or
$\big[b_1= b'_1=0$ and $a_1\neq a'_1\big]$.

\item[(ii)] If $n=m>1$, then $T\parallel T'$ if and only if 
\begin{itemize}
\item[a)] $b_n\neq b'_n$, or 
\item[b)] $b_n= b'_n=0$ and $a_n\neq a'_n$, or
\item[c)] $b_n= b'_n\neq 0$ and 
there exists $1\leqslant j<n$ such that $(a_j,b_j)\neq (a'_j,b'_j)$, or 
\item[d)] $b_n= b'_n= 0$, $a_n= a'_n\neq 0$, and $a_{n-1}\neq a'_{n-1}$ or there exists $1\leqslant j<n-1$ such that $(a_j,b_j)\neq (a'_j,b'_j)$.
\end{itemize}
\item[(iii)] If $n\neq m$, then $T\parallel T'$.
\end{itemize}
\end{lemma}

\begin{proof}
We may assume that $a_i\neq 0$ for $2\leqslant i\leqslant n$ and $b_i\neq 0$ for $1\leqslant n-1$,
and that $a'_j\neq 0$ for $2\leqslant j\leqslant m$ and $b'_j\neq 0$ for $1\leqslant j\leqslant m-1$.

{\bf (i):} To prove sufficiency, suppose first that  $b_1\neq b'_1$, say $b_1>b'_1$.
Then, by considering  $\sigma=(1,1)(0,1)^{b'_1}(1,1)$ and $\sigma'=(1,1)(0,1)^{b'_1+1}$, we see that
$T\not \leqslant T'$ and $T\not \geqslant T'$, respectively.

So suppose that $b_1= b'_1=0$ and $a_1\neq a'_1$, say $a_1>a'_1$. Then,
by considering  $\sigma=(1,1)(1,0)^{a'_1}(1,1)$ and $\sigma'=(1,1)(1,0)^{a'_1+1}$, we see that
$T\not \leqslant T'$ and $T\not \geqslant T'$, respectively. 

We prove necessity by counterposition. Observe first that if $b_1= b'_1$ and  $(a_1,b_1)=(a'_1,b'_1)$, then  $T=T'$.
So suppose that $b_1= b'_1\neq 0$ and $a_1\neq a'_1$ , say $a_1>a'_1$. We claim that $T<T'$.
By making use of $\sigma=(1,1)(1,0)^{a'_1+1}$, we see that $T(\sigma)\sqsubset T'(\sigma)$.
Thus, we only have to show that $T\leqslant T'$.

Let $\sigma=(p_1,m_1)(p_2,m_2)\cdots (p_k,m_k)$. 
If the action of $\rho_1^a \rho_2^b \rho_3$ does not delete all terms of $(p_1,m_1)$, then 
 $T(\sigma)\sqsubseteq T'(\sigma)$ since the ``subrule" $\langle \cdot \rangle^+_-$ in $T$ and $T'$ does not act on $\sigma$; 
hence, without loss of generality, we may further assume that $(p_1,m_1)=(1,1)$.
 
 
Now, if $m_2\neq 0$, then $T(\sigma)\sqsubseteq T'(\sigma)$, and if $p_2=0$,
then $T(\sigma)= T'(\sigma)$; hence, we may assume $m_2=0$ and
$p_2\neq 0$.  In fact, we may suppose that $p_2\leqslant a'_1$ since, otherwise,
$T(\sigma)\sqsubset T'(\sigma)$.
%

Under the assumption that $m_2=0$ and $p_2\leqslant a'_1$, and  applying the same
reasoning to $(p_3,m_3)$, we again derive that the only case to consider is when
$m_3=0$ and $0<p_3\leqslant a'_1-p_2$. Proceeding in this way, we may eventually arrive at
$(p_j,m_j)$ with $m_i=0$ and $p_i>0$ for $i=2,\ldots,j-1$, $m_j=0$ and 
\[
a'_1-\sum_{i=2}^{j-1}p_i\leqslant 0.
\] 
The only case to consider reduces then to $p_j=0$, hence $(p_j,m_j)=(0,0)$, so
this term disappears, and similarly all remaining terms till $(p_k,m_k)$. Otherwise, if
\[
a'_1-\sum_{i=2}^{k-1}p_i>0
\]
we have to consider the case $(p_k,m_k)=(p_k,0)$ with 
$0<p_k\leqslant a'_1-\sum_{i=2}^{k-1}p_i$. But then clearly $T(\sigma)\sqsubset T'(\sigma)$. 
In any case we have $T\leqslant T'$ (hence, $T\not\parallel T'$),  and the proof is now complete.

{\bf (ii):} The proof of sufficiency in the case when   $(a_j,b_j)= (a'_j,b'_j)$, for $1\leqslant j<n$, 
follows exactly the same steps as in the proof of (i), by adding $(0,1)[(1,0)(0,1)]^{n-2}$ (after the first $(1,1)$) to the strings used above. 
We consider the case when $b_n= b'_n\neq 0$ and there exists $1\leqslant j<n$ such that $(a_j,b_j)\neq (a'_j,b'_j)$, say $a_j>a'_j$. 
The case $b_n= b'_n= 0$,  $a_{n-1}= a'_{n-1}\neq 0$, and $(a_j,b_j)\neq
(a'_j,b'_j)$ (say $b_j>b'_j$) for some $1\leqslant j<n-1$, 
follows similarly by interchanging the roles of $\rho_4$ and $\rho_5$, and those of $(1,0)$ and $(0,1)$ (and $a'_j$ and $b'_j$) in the strings below.

So let $\sigma $ be given by
\begin{itemize}
\item $\sigma=(1,1)(0,1)(1,1)^{j-2}(1,0)^{a'_j+1}(0,1)(1,1)^{n-j-1}(1,0)$ if $1<j$, and 
\item $\sigma=(1,1)(1,0)^{a'_j+1}(0,1)(1,1)^{n-1}(1,0)$ otherwise. 
\end{itemize}
Then $T(\sigma)\neq \varepsilon=T'(\sigma)$ and thus $T\not\leqslant T'$.
Now let $\sigma'$ be given by 
\begin{itemize}
\item $\sigma'=(1,1)(0,1)(1,1)^{j-2}(1,0)^{a'_j+1}(0,1)(1,1)^{n-j-2}(1,0)$ if $1<j$ 
(with $(1,1)^{n-j-2}=(0,0)$ whenever $n-j-2\leqslant 0$), and 
\item $\sigma'=(1,1)(1,0)^{a'_j+1}(0,1)(1,1)^{n-2}(1,0)$ otherwise. 
\end{itemize}
Then $T(\sigma')= \varepsilon\neq T'(\sigma')$ and thus $T\not\geqslant T'$.

The proof of necessity is similar to case (i). If none of the conditions of (ii) is satisfied, then we may assume that
$b_n= b'_n\neq 0$,  $(a_j,b_j)= (a'_j,b'_j)$ for every $1\leqslant j<n$, and focus on the case
$a_n\neq a'_n$ (for otherwise $T=T'$).
(The case when $b_n= b'_n= 0$, $a_n=a'_n$, $b_{n-1}=b'_{n-1}$, and $(a_j,b_j)= (a'_j,b'_j)$ for every $1\leqslant j<n-1$ follows 
similarly by the above mentioned substitutions.)

So suppose without loss of generality that  $a_n>a'_n=t$.
As in case (i), we show that $T<T'$. 
By making use of 
$\sigma=(1,1)(0,1)(1,1)^{n-2}(1,0)^{a'_n+1}$, we see that $T(\sigma)\sqsubset T'(\sigma)$.

So let $\sigma=(p_1,m_1)(p_2,m_2)\cdots (p_k,m_k)$. By reasoning as in (i), we may assume
that $k>t+n$ and, since  $(a_j,b_j)= (a'_j,b'_j)$ for every $1\leqslant j<n$, that $(p_2,m_2)=(0,1)$, 
$(p_j,m_j)=(1,1)$ for  $3\leqslant j<n+1$, and that $p_j= 0$ for each $n+1\leqslant j \leqslant t+n$;
for otherwise we reach the same conclusion $T(\sigma)\sqsubseteq T'(\sigma)$.

If $p_{t+n+1}= 0$ or $m_{t+n+1}= 0$,
 then $T(\sigma)= T'(\sigma)$ or $T(\sigma)\sqsubset T'(\sigma)$, respectively.
 If  $p_{t+n+1}\neq 0$ and $m_{t+n+1}\neq 0$, then $T(\sigma)\sqsubseteq T'(\sigma)$,
 and the proof of (ii) is now complete.

{\bf (iii):}  Suppose that  $n\neq m$, say $1\leqslant n<m$. First we consider the case $n=1$.
Let $\sigma=(1,1)(0,1)(1,0)$. Then $T(\sigma)\neq \varepsilon =T'(\sigma)$ and thus $T\not\leqslant T'$.
Let $\sigma'=(1,1)(0,1)^\alpha(1,1)^{m-1}(1,0)$ where $\alpha=0$ if $b_1=0$, and $\alpha=1$ otherwise. Then $T(\sigma')=\varepsilon\neq T'(\sigma')$ and thus $T\not\geqslant T'$.

Suppose now that $n>1$. Then, for $\sigma=(1,1)(0,1)(1,1)^{m-2}(1,0)$, we have $T(\sigma)\neq \varepsilon=T'(\sigma)$ and thus $T\not\leqslant T'$.

To show that $T\not\geqslant T'$, let $\sigma'$ be given by
\begin{itemize}
\item $\sigma'=(1,1)(0,1)(1,1)^{n-1}(1,1) (1,1)^{m-n-1}(1,0)$ if $b_n,b'_m\neq 0$,
\item $\sigma'=(1,1)(0,1)(1,1)^{n-2}(1,0)(1,1)(0,1) (1,1)^{m-n-1}(1,0)$ if $b_n=0$ and $b'_m\neq 0$,
\item $\sigma'=(1,1)(0,1)(1,1)^{n-1}(1,1) (1,1)^{m-n-2}(1,0)(0,1)$ if $b_n\neq 0 $ and $b'_m= 0$, and 
\item $\sigma'=(1,1)(0,1)(1,1)^{n-2}(1,0)(1,1)(0,1)  (1,1)^{m-n-1}(1,0)$ if $b_n,b'_m= 0$.
\end{itemize}
In each case we get $T(\sigma')=\varepsilon\neq T'(\sigma)$. Thus $T\not\geqslant T'$, and the proof of Lemma~\ref{lem:3} is now complete.
\end{proof}

\begin{remark}\label{Rem:1234}
Note that the proofs of (i) and (ii) of Lemma~\ref{lem:3} show that if $b_n=b'_n\neq 0$ and 
$(a_j,b_j)= (a'_j,b'_j)$ for $1\leqslant j<n$, then $T<T'$ if and only if $a_n>a'_n$. Similarly,
if $b_n=b'_n= 0$, $a_n=a'_n\neq 0$, $a_{n-1}= a'_{n-1}$, and 
$(a_j,b_j)= (a'_j,b'_j)$ for $1\leqslant j<n-1$, then $T<T'$ if and only if $b_{n-1}>b'_{n-1}$. 

Moreover, if 
$b_n=b'_n$, $a_n=a'_n$, and 
there exists $1\leqslant j< n$ such that  $a_j>a'_j$ or $b_j>b'_j$, then we have that
 $\omega'\rho_1\rho_2\rho_3\langle \cdot \rangle^+_-\not\leqslant 
\omega\rho_1^a\rho_2^b\rho_3T$ for every 
$T\in \gR$ and any $a,b\in \{0,1\}$. To illustrate, suppose $a_j>a'_j$. Then
consider $\sigma $ given by
\begin{itemize}
\item $\sigma=(1,1)(0,1)(1,1)^{j-2}(1,0)^{a'_j+1}(0,1)[(1,0)(0,1)]^{n-j-1}(1,0)$ if $1<j$, and 
\item $\sigma=(1,1)(1,0)^{a'_j+1}(0,1)[(1,0)(0,1)]^{n-1}(1,0)$ otherwise. 
\end{itemize}
\end{remark}

\begin{remark}\label{Rem:123456} 
By reasoning as in the proof of (iii) of Lemma~\ref{lem:3} and taking $\omega $ and $\omega'$ as above with $m<n$,
one can show that $\omega'\rho_1\rho_2\rho_3\langle \cdot \rangle^+_-\not\leqslant 
\omega\rho_1^a\rho_2^b\rho_3T$  for every $T\in \gR$ and any $a,b\in \{0,1\}$.
\end{remark}

\subsection{The subposet $\gR_{123}$.}
Let $\gR_{123}:=\{R\in\gR/_\sim:R\in \cL(\{\rho_1,\rho_2,\rho_3\})\}.$ Writing
these rules in factorized irredundant form, they read $R=T_1T_2\cdots$ where, for
each $i\in \mathbb{N}$, $T_i=\rho_1^a \rho_2^b \rho_3$ for some
$a,b\in\{0,1\}$. Note that from Proposition~\ref{prop:ab} it follows that each such expression is
necessarily in $\gR/_\sim$, is in the factorized irredundant form, and 
has infinite length. 

For $T\in \gR_{123}$, and $a,b\in\{0,1\}$, set $\mathcal{I}^T_{ab}=\{i\in
\mathbb{N}:T_i=\rho_1^a \rho_2^b \rho_3\}$.  Since $T$ is of infinite length,
$(\mathcal{I}^T_{ab})_{a,b\in\{0,1\}}$ is a partition of $\mathbb{N}$. Moreover, 
$(\mathcal{I}^T_{ab})_{a,b\in\{0,1\}}$ uniquely determines $T$.

\begin{proposition}\label{prop:1}
Let $T,T'\in \gR_{123}$. Then $T\leqslant T'$ if and only if $\mathcal{I}^T_{11}\supseteq \mathcal{I}^{T'}_{11}$ and 
$\mathcal{I}^T_{ab}\subseteq \mathcal{I}^{T'}_{ab}$, for any $(a,b)\neq(1,1)$.
In particular, 
\begin{itemize}
\item[(i)] $T<T'$ whenever $\mathcal{I}^T_{11}\supset \mathcal{I}^{T'}_{11}$ and 
$\mathcal{I}^T_{ab}\subseteq \mathcal{I}^{T'}_{ab}$, for any $(a,b)\neq(1,1)$;
\item[(ii)] $T\parallel T'$ whenever $\mathcal{I}^T_{ab}\parallel \mathcal{I}^{T'}_{ab}$ for some $a,b\in\{0,1\}$.
\end{itemize}
\end{proposition}

\begin{proof} Clearly, the third claim is a consequence of the first. 
Since $(\mathcal{I}^T_{ab})_{a,b\in\{0,1\}}$ is a partition of $\mathbb{N}$, 
the second claim is also a consequence of the first.

To see that the conditions in the first claim are sufficient, note that, if $T$
acts on a string $\sigma =(p_1,m_1)\cdots (p_k,m_k)$, then its factor $T_i$ acts
on the term $(p_i,m_i)$. Suppose that for some $\sigma$ we
have $T(\sigma)\sqsupset T'(\sigma)$. Then, using the above remark, for some $i$
we have  $T_i(p_i,m_i)\sqsupset T'_i(p_i,m_i)$,
 which means that $T'_i=\rho_1^a
\rho_2^b \rho_3$, $T_i=\rho_1^c \rho_2^d \rho_3$ with $(c,d)\leqslant (a,b)$
pointwise and $(c,d)\neq (a,b)$.  There are two possibilities: 
\begin{itemize}
\item $(a,b)=(1,1)$, and hence $\I^T_{11}\not\supseteq \I^{T'}_{11}$, or 
\item $(a,b)=(1,0)$ or $(a,b)=(0,1)$, in which case $(c,d)=(0,0)$, and thus  $\I^T_{00}\not\subseteq \I^{T'}_{00}$.
\end{itemize}

To show that the conditions of the first claim are necessary, 
suppose first that $\mathcal{I}^T_{11}\not \supseteq \mathcal{I}^{T'}_{11}$.
Let $i=\min\{j\in \mathcal{I}^{T'}_{11}: j\not \in\mathcal{I}^T_{11}\}$. 
If $i\in \mathcal{I}^T_{ab}$ for $ab= 00$ or $01$, consider 
$\sigma=(1,1)^{i-1}(2,1)$. Then $T(\sigma)=(1,0)\neq \varepsilon=T'(\sigma)$.
If $i\in \mathcal{I}^T_{10}$, consider 
$\sigma=(1,1)^{i-1}(1,2)$. Then $T(\sigma)=(0,1)\neq\varepsilon=T'(\sigma)$.
Thus $T\not\leqslant T'$.

So we may assume that  $\mathcal{I}^T_{11}\supseteq \mathcal{I}^{T'}_{11}$.
We treat the case $\mathcal{I}^T_{01}\not \subseteq \mathcal{I}^{T'}_{01}$;
the remaining cases follow similarly.
Let $i=\min\{j\in \mathcal{I}^{T}_{01}: j\not \in\mathcal{I}^{T'}_{01}\}$. 
If $i\in \mathcal{I}^{T'}_{10}$, consider 
$\sigma=(1,1)^{i-1}(2,1)$. Then $T(\sigma)=(1,0)\neq\varepsilon=T'(\sigma)$.
If $i\in \mathcal{I}^{T'}_{00}$, consider 
$\sigma=(1,1)^{i-1}(2,2)$. Then $T(\sigma)=(1,0)\neq\varepsilon=T'(\sigma)$.
Thus $T\not\leqslant T'$, and the proof is now complete.
\end{proof}

As immediate corollaries we have the following results.

\begin{corollary}\label{cor:1}
Let $T\in \gR_{123}$. 
\begin{itemize}
\item[(i)] $T$ is the least rule if and only if $\mathcal{I}^T_{11}=\mathbb{N}$,
  i.e., $T=\langle\cdot\rangle_=$.
\item[(ii)] $T$ is an atom if and only if $\mathcal{I}^T_{11}=\mathbb{N}\setminus \{i\}$ for some 
$i\in \mathbb{N}$.
\item[(iii)] $T$ is a maximal element if and only if $\mathcal{I}^T_{11}=\emptyset$.
\end{itemize}
\end{corollary}

\begin{corollary}\label{cor:2}
Let $T, T'\in \gR_{123}$ where $T=T_1T_2\cdots$ and $T'=T'_1T'_2\cdots$.
Then $T\wedge T'=S$ where $\mathcal{I}^S_{11}=\mathcal{I}^T_{11}\cup \mathcal{I}^{T'}_{11}\cup \bigcup_{(a,b)\neq(1,1)}\mathcal{I}^T_{ab}\oplus \mathcal{I}^{T'}_{ab}$,
where $\oplus$ denotes the symmetric difference, and 
$\mathcal{I}^S_{ab}=\mathcal{I}^T_{ab}\cap \mathcal{I}^{T'}_{ab}$ for every $(a,b)\neq(1,1)$.
\end{corollary}

In other words,  $\gR_{123}$ constitutes a $\wedge$-semilattice. 
Now, by Proposition \ref{prop:1}, if $T ,T' \leqslant R\in \gR_{123}$, then 
$\mathcal{I}^T_{11}, \mathcal{I}^{T'}_{11}\supseteq \mathcal{I}^{R}_{11}$
and $\mathcal{I}^T_{ab}, \mathcal{I}^{T'}_{ab}\subseteq \mathcal{I}^{R}_{ab}$
 for every $(a,b)\neq(1,1)$.
 Hence, Corollary \ref{cor:2} can be refined by considering intervals of the form 
 $[\langle \cdot\rangle_=,R]$ for some  $R\in \gR_{123}$.

\begin{corollary}\label{cor:3}
Let $R\in \gR_{123}$. Then $([\langle \cdot\rangle_=,R], \leqslant)$ constitutes a lattice under $\wedge$ and $\vee$
defined by
\begin{enumerate}
\item $T\wedge T'=S$ where $\mathcal{I}^S_{11}=\mathcal{I}^T_{11}\cup \mathcal{I}^{T'}_{11}\cup \bigcup_{(a,b)\neq(1,1)}\mathcal{I}^T_{ab}\oplus \mathcal{I}^{T'}_{ab}$, and 
$\mathcal{I}^S_{ab}=\mathcal{I}^T_{ab}\cap \mathcal{I}^{T'}_{ab}$ for every $(a,b)\neq(1,1)$; 
\item $T\vee T'=S$ where $\mathcal{I}^S_{11}=\mathcal{I}^T_{11}\cap \mathcal{I}^{T'}_{11}$, and 
$\mathcal{I}^S_{ab}=\mathcal{I}^T_{ab}\cup \mathcal{I}^{T'}_{ab}$ for every $(a,b)\neq(1,1)$,
\end{enumerate}  for every $T ,T' \in [\langle \cdot\rangle_=,R]$, with $T=T_1T_2\cdots$ and $T'=T'_1T'_2\cdots$.
Moreover, $([\langle \cdot\rangle_=,R], \leqslant)$ is order-isomorphic to 
$(\mathcal{P}(\bigcup_{(a,b)\neq(1,1)}\mathcal{I}^R_{ab}),\subseteq)$.
\end{corollary}

From Corollary~\ref{cor:3}, it follows that $(\gR_{123}, \leqslant)$ embeds the
power set of integers ordered by inclusion.
Furthermore, for $R\in \gR_{123}$, if $ |\bigcup_{(a,b)\neq(1,1)}\mathcal{I}^R_{ab}|=n$ is finite, then $|[\langle \cdot\rangle_=, R]|=2^n$.

\subsection{Least element and atoms}\label{atoms}
We turn to the study of the atoms of $\gR/_\sim$. The next proposition was presented
in \cite{Gra03}.
\begin{proposition}\label{prop:3}
The rule $\langle \cdot\rangle_-^+$ is the least element of $\gR/_\sim$.
\end{proposition}
\begin{proof}
It follows immediately from the fact that $\langle \cdot\rangle_-^+$
deletes every term of a nonassociative string.
\end{proof}

\begin{proposition}\label{prop:3a}
 Let $T=\omega \rho_1^a \rho_2^b \rho_3 T' $ be an element of
 $\gR/_\sim$, where $(a,b)\neq(1,1)$. Then $T$ is an atom if and only if $\rho_4$ and
 $\rho_5$ alternate infinitely many times in $\omega$ (and therefore $T'=\varepsilon$).
\end{proposition}
\begin{proof}
Note that if $\rho_4$ and $\rho_5$ alternate infinitely many times in $\omega$, then 
$T=\omega \rho_1^a \rho_2^b \rho_3 $. By Lemma \ref{lem:1}, the condition is sufficient.

To show that it is also necessary, let $T$ be an atom, and for the sake of a contradiction
 suppose that 
$\rho_4$ and $\rho_5$ alternate finitely many times in $\omega$.
Let $T''=(\rho_4\rho_5)^*T$. Then  $T''=(\rho_4\rho_5)^* \rho_1^a \rho_2^b \rho_3 > \langle \cdot\rangle_-^+ $.
Moreover, by Lemma \ref{lem:2} we have $T>T''$ which contradicts the fact that $T$ is an atom. 
\end{proof}
Consequently, for $(a,b)<(1,1)$, we have only 3 atoms, namely
$(\rho_4\rho_5)^*\rho_3$, $(\rho_4\rho_5)^*\rho_1\rho_3$ and $(\rho_4\rho_5)^*\rho_2\rho_3$.

\begin{proposition}\label{prop:4}
If $T=\omega\rho_1 \rho_2 \rho_3T'$ is an atom, then $\rho_4$ and $\rho_5$ alternate finitely many times in $\omega$.
\end{proposition}

\begin{proof}
If $\rho_4$ and $\rho_5$ alternate infinitely many times in $\omega$, then $T=\langle \cdot\rangle_-^+$.
\end{proof}

\begin{proposition}\label{prop:5}
Let $T=\omega\rho_1 \rho_2 \rho_3T'$ such that $\rho_4$ and $\rho_5$ alternate finitely many times in $\omega$.
If $T$ is an atom, then $T'=\langle \cdot\rangle_-^+$.
\end{proposition}

\begin{proof} Let $T=\omega\rho_1 \rho_2 \rho_3T'$ be an atom.
Suppose for the sake of a contradiction that $T'\neq\langle \cdot\rangle_-^+$.
By Proposition \ref{prop:3}, $T'>\langle \cdot\rangle_-^+$, and by Lemma \ref{lem:2}, 
 $T''=\omega\rho_1 \rho_2\rho_3 \langle \cdot\rangle_-^+>\langle \cdot\rangle_-^+$.
 Moreover, by Lemma \ref{lem:4}, $T>T''$ which contradicts the fact that $T$ is an atom.
\end{proof}

\begin{proposition}\label{prop:6}
Let $T=\omega\rho_1 \rho_2 \rho_3 \langle \cdot\rangle_-^+$. Then $T$ is an atom if and only if $\omega:=\rho_4^{a_1}\rho_5^{b_1}\cdots\rho_4^{a_n}\rho_5^{b_n}$ with $a_i\neq 0$ for $2\leqslant i\leqslant n$ and $b_i\neq 0$ for $1\leqslant i\leqslant n-1$, and such that 
\begin{itemize}
\item[(i)] $b_n\neq 0$ and $a_n$ is infinite, or
\item[(ii)] $b_n=0$,  $a_n\neq 0$ and $b_{n-1}$ is infinite. 
\end{itemize}
\end{proposition}

\begin{proof}  Necessity follows from Propositions \ref{prop:4} and \ref{prop:5}, and Lemma~\ref{lem:3} and Remark~\ref{Rem:1234}.
Sufficiency follows from Lemma~\ref{lem:3} and Remark~\ref{Rem:1234}.
\end{proof}

We can now explicitly describe the atoms of $\gR/_\sim$. 
\begin{theorem}
 A w.f.c.r. $T$ is an atom of $\gR/_\sim$ if and only if $T=(\rho_4\rho_5)^* \rho_1^a \rho_2^b \rho_3$, for $(a,b)\neq(1,1)$, or 
 $T=\omega\rho_1 \rho_2 \rho_3 \langle  \cdot\rangle_-^+$ where $\omega:=\rho_4^{a_1}\rho_5^{b_1}\cdots\rho_4^{a_n}\rho_5^{b_n}$ with $a_i\neq 0$ for $2\leqslant i\leqslant n$ and $b_i\neq 0$ for  $1\leqslant i\leqslant n-1$, and  such that 
\begin{itemize}
\item[(i)] $b_n\neq 0$ and $a_n$ is infinite, or
\item[(ii)] $b_n=0$,  $a_n\neq 0$ and $b_{n-1}$ is infinite. 
\end{itemize}
\end{theorem}
  
\subsection{Maximal elements}\label{maximals}
We  now focus on the maximal elements of $\gR/_\sim$. In \cite{Gra03}, it was proved that
$\langle\cdot\rangle_0$ is a maximal element of the set of well-formed
computation rules. 
 
\begin{proposition}\label{prop:max1}
 Let $T=\omega \rho_1^a \rho_2^b \rho_3T'$. 
 If $T$ is maximal, then
\begin{itemize}
\item[(i)] $\rho_4$ and $\rho_5$ alternate finitely many times in $\omega$,
\item[(ii)] $(a,b)<(1,1)$, and 
\item[(iii)] $T'$ is maximal.
\end{itemize}
\end{proposition}
 
\begin{proof}
Conditions $(i)$ and $(ii)$ follow from Lemmas \ref{lem:2} and \ref{lem:1}, respectively.
Condition $(iii)$ follows from Lemma \ref{lem:4}.
\end{proof}
 
As it turns out, every maximal element of $\gR_{123}$ is also maximal in $\gR/_\sim$.
\begin{proposition}\label{prop:max123}
Let $T\in \gR_{123}$. If $\mathcal{I}^T_{11}=\emptyset$, then
$T$ is a maximal element of $\gR/_\sim$.
\end{proposition} 
\begin{proof}
 It suffices to show that for every $T'=T'_1T'_2\cdots $ such that $T'\geqslant
 T=T_1T_2\cdots$, we have $T'\in \mathcal{R}_{123}.$

 For the sake of a contradiction, suppose $T'\not \in \mathcal{R}_{123}$, and
 let $i=\min\{j\in \mathbb{N}: T'_i\not \in \cL(\{\rho_1, \rho_2, \rho_3\})\}$.
 Note that $T_i=\rho_1^a \rho_2^b \rho_3$ for $(a,b)<(1,1)$, for every $i\in
 \mathbb{N}$.  Since $T'\geqslant T$, $T'_i=\omega \rho_1^a \rho_2^b \rho_3$
 where $\rho_4$ and $\rho_5$ alternate finitely many times in $\omega$.  Without
 loss of generality, suppose $\omega=\rho_4\omega'$.  Consider $\sigma
 =(1,1)^{i-1}(1,0)^2$. Then $T(\sigma)=(1,0)^2>(1,0)=T'(\sigma)$, thus yielding
 the desired contradiction.

Hence, $T'\in  \mathcal{R}_{123}$ and, by Corollary \ref{cor:1}, $T'=T$.
Thus $T$ is maximal in $\mathcal{R}$.
\end{proof}
 
Now we consider the maximal elements $T\in (\gR/_\sim)\setminus \gR_{123}$.
 
\begin{proposition}\label{prop:notmax1}
Let $T=\omega\rho_1^a \rho_2^b \rho_3R$. If $\omega \not \in \cL(\rho_4)\cup\cL(\rho_5)$,
then $T$ is not maximal.
\end{proposition} 
\begin{proof}
 Let $\omega = \rho_4^{a_1}\rho_5^{b_1}\cdots \rho_4^{a_n}\rho_5^{b_n}$,
 $n\geqslant 1$. 
 Without loss of generality, suppose that $a_i\neq 0$, $b_i\neq 0$, for every $1\leqslant i\leqslant n$.
  Assume also that $a,b=0$; the other cases
 $(a,b)<(1,1)$ follow similarly.
 
 Let $R'=(\rho_1 \rho_2 \rho_3)^nR$, and set $T'=\rho_1^a \rho_2^b \rho_3R'$. 
 Let $\gamma=(1,1)(1,0)(0,1)$. Then
 $T(\gamma)=\varepsilon<(1,0)(0,1)=T'(\sigma)$, and thus $T<T'$. Hence, $T$ is not maximal. 
\end{proof}
 
 As an immediate corollary we get the following necessary condition for maximality.
 
\begin{corollary}\label{cor:max4or5}
Let $T=T_1T_2\cdots$ with $T_i=\omega_i\rho_1^{a_i} \rho_2^{b_i} \rho_3$.
If $T$ is maximal, then for every $i\in \mathbb{N}$, $\omega_i\in \cL(\rho_4)\cup\cL(\rho_5)$.
\end{corollary}

\begin{remark}
Let $T=T_1T_2\cdots$ where $T_i=\omega_i\rho_1^{a_i} \rho_2^{b_i} \rho_3$, with
$\omega_i \in \cL(\rho_4)\cup\cL(\rho_5)$.  If for some $i\in \mathbb{N}$,
$\omega_i$ is $\rho_4^*$, then $T=T_1\cdots T_i$. Otherwise,
$T=T_1T_2\cdots$, and for each $i$ there is a string $\gamma$ such that $T_i$
acts on $\gamma$.
\end{remark}

\begin{proposition}\label{prop:7}
Let $T=T_1T_2\cdots$  where $T_i=\omega_i\rho_1^{a_i} \rho_2^{b_i} \rho_3$,  and $T'=T'_1T'_2\cdots $  where $T'_i=\omega'_i\rho_1^{a_i} \rho_2^{b_i} \rho_3$, with $\omega_i, \omega'_i \in \cL(\rho_4)\cup\cL(\rho_5)$. Then $T\parallel T'$ if and only if one of the following holds:
  \begin{itemize}
 \item[(i)] each $\omega_i, \omega'_i$ has finite length, and $\omega_i\neq \omega'_i$, for some $i\in \mathbb{N}$, 
 \item[(ii)] $T=T_1\cdots T_i$ and neither $\rho_4^*$ nor $\rho_5^*$ occur in $T_j$, $1\leqslant j\leqslant i-1$, nor in $T'$.
 \item[(iii)] $T=T_1\cdots T_i$ and $T'=T'_1\cdots T'_j$ where neither $\rho_4^*$ nor $\rho_5^*$ 
 occur in $T_l$, $1\leqslant l\leqslant i-1$, nor in $T'_k$, $1\leqslant k\leqslant j-1$, or $\omega_t\neq \omega'_t$, for some 
 $1\leqslant t\leqslant i\wedge j$.
  \end{itemize}
  \end{proposition}

\begin{proof} To see that the conditions in (i)-(iii) are necessary, observe that if  $T\parallel T'$ ($T$ and $T'$ in factorized irredundant forms), then we must have 
$T\neq T'$. Since each $\omega_i$ and each $\omega'_i$ is in $\cL(\rho_4)\cup\cL(\rho_5)$,  
one of (i)-(iii) must occur. 

To show that (i) is sufficient, assume that each $\omega_i, \omega'_i$ has finite length and, without loss of generality,
suppose that  $\omega_1\neq \omega'_1$. We consider 3 representative cases; 
the remaining cases follow similarly.

Suppose that  $\omega_1\in \cL(\{\rho_4\})$ and $\omega'_1\in\cL(\{\rho_5\})$.
Take $\sigma=(1,1)(0,1)$ and $\sigma'=(1,1)(1,0)$.
Then $T(\sigma)=(0,1)\neq\varepsilon=T'(\sigma)$, but $T(\sigma')=\varepsilon \neq (1,0)=T'(\sigma')$.

Suppose that $\omega_1\in \cL(\{\rho_4\})$ and $\omega'_1=\varepsilon$.
Take $\sigma=(1,1)(1,0)$ and $\sigma'=(1,1)(1,1)$.
Then $T(\sigma)=\varepsilon\neq(1,0)=T'(\sigma)$, but $T(\sigma')=(0,1)\neq\varepsilon =T'(\sigma')$.

Suppose now that $\omega_1\in \rho_4^n$ and $\omega'_1=\rho_4^m$, 
say, $n<m$. Take $\sigma=(1,1)(1,0)^{n+1}$ and $\sigma'=(1,1)(1,0)^n(1,1)$.
Then $T(\sigma)=(1,0)\neq\varepsilon=T'(\sigma)$, but $T(\sigma')=\varepsilon \neq(0,1)=T'(\sigma')$.

In all representative cases we conclude that $T\parallel T'$.

To show that (ii) is sufficient, suppose that $T=T_1\cdots T_i$ and neither $\rho_4^*$ nor $\rho_5^*$ occur in $T_l$, $1\leqslant l\leqslant i-1$, 
nor in $T'$. Let $k=\min\{j:\omega_j\neq \omega'_j\}$.
If $k\leqslant i-1$, then the proof of (i) can be used to show that $T\parallel T'$. 

So suppose that $k=i$ and, without loss of generality, suppose that $\omega_i= \rho_4^*$ and $\omega'_i=\rho_4^m$, $m>0$.
Take $\sigma=(1,1)^i(1,0)^{m+1}$ and $\sigma'=(1,1)(1,0)^m(1,1)$.
Then $T(\sigma)=\varepsilon\neq (1,0)=T'(\sigma)$, but $T(\sigma')=(0,1)\neq\varepsilon =T'(\sigma')$,
and again we have that $T\parallel T'$.

Finally, to show that (iii) is sufficient, suppose that 
$T=T_1\cdots T_i$ and $T'=T'_1\cdots T'_j$ where neither $\rho_4^*$ nor $\rho_5^*$ occur in $T_l$, $1\leqslant l\leqslant i-1$, 
nor in $T'_k$, $1\leqslant k\leqslant j-1$, or $\omega_t\neq \omega'_t$, for some 
 $1\leqslant t\leqslant i\wedge j$.
 
 Now, as case (i), we may assume that $i<j$ (the case $i>j$ is similar), and that 
 $\omega_i\in \cL(\{\rho_4\})$ and $\omega'_i=\rho_4^m$, $m>0$. But then, as in case (ii),
we again have $T\parallel T'$, and thus the proof is now complete.
\end{proof}

From Lemma \ref{lem:1}, the above necessary condition and Propositions \ref{prop:max123} and \ref{prop:7}, we obtain the following
explicit description of the maximal elements of $\gR/_\sim$.
\begin{theorem}
Let $T\in \gR/_\sim$. Then $T$ is maximal if and only if
 \begin{itemize}
 \item[(i)] $T$ is a maximal element of $\gR_{123}$, or 
 \item[(ii)] $T=T_1T_2\cdots$  where $T_i=\omega_i\rho_1^{a_i} \rho_2^{b_i} \rho_3$ with $\omega_i \in \cL(\rho_4)\cup\cL(\rho_5)$ and $(a_i,b_i)<(1,1)$.
  \end{itemize}
\end{theorem}

\section{Concluding remarks: An alternative ordering of $\gR/_\sim$}

An alternative ordering of $\gR/_\sim$ was proposed in \cite{Gra03}, and which is defined as follows.
Given $R\in \gR/_\sim$, let $\mathrm{Ker}(R):=\{\sigma: R(\sigma)=\varepsilon\}.$ 
For $R, R'\in \gR/_\sim$, we write $R\leqslant_{\mathrm{Ker}}R'$ if  $\mathrm{Ker}(R)\supseteq \mathrm{Ker}(R')$.
Clearly, $\leqslant_{\mathrm{Ker}}$ is a partial ordering of $\gR/_\sim$, and
 if $R\leqslant R'$, then $R\leqslant_{\mathrm{Ker}}R'$; see  \cite{Gra03}. 
 As it turns out, the converse is also true.

\begin{proposition}\label{prop:equivOrders}
Let $R, R'\in \gR/_\sim$. Then $R\leqslant R'$ if and only if $R\leqslant_{\mathrm{Ker}}R'$.
\end{proposition}
 
\begin{proof} 
To prove Proposition \ref{prop:equivOrders} it remains to show that if $R\parallel R'$, then 
$R\parallel_{\mathrm{Ker}}R'$, i.e., $R\not\leqslant_{\mathrm{Ker}}R'$ and $R'\not\leqslant_{\mathrm{Ker}}R$.

So suppose that  $R\parallel R'$, that is, there exist $\sigma_1$ and $\sigma_2$ such that
$R(\sigma_1) \sqsubset R'(\sigma_1)$ and $R(\sigma_2) \sqsupset R'(\sigma_2)$.

Let $ \sigma'_1$ the string be obtained from $ \sigma_1$ by removing 
the indices in $R(\sigma_1)$ such that $R(\sigma'_1)=\varepsilon \neq R'(\sigma'_1)$.
Hence, $R'\not\leqslant_{\mathrm{Ker}}R$.

Similarly, let $ \sigma'_2$ the string be obtained from $ \sigma_2$ by removing 
the indices in $R'(\sigma_2)$ such that $R(\sigma'_2) \neq \varepsilon=R'(\sigma'_2)$.
Hence, $R\not\leqslant_{\mathrm{Ker}}R'$.

Thus $R\parallel_{\mathrm{Ker}}R'$, and the proof is now complete.
\end{proof}

We have presented a partial description of the poset $\gR/_\sim$; being uncountable, there is little hope 
of obtaining an explicit description as it was the case of the subposet $\mathcal{R}_{123}$, which was shown
to be isomorphic to the power set of natural numbers.

Looking at directions of further research, we are inevitably drawn to the question in 
determining whether $\gR/_\sim$ constitutes a $\wedge$-semilattice and, if that is the case,
whether its closed intervals constitute lattices, as it was the case  of the subposet $\mathcal{R}_{123}$.

\end{document}